\documentclass[11pt]{article}

\usepackage[margin=1in]{geometry}
\usepackage[cmex10]{amsmath}
\usepackage{enumerate}
\usepackage{graphicx}
\usepackage{amsfonts}
\usepackage{amsthm}
\usepackage{url}
\newcommand{\mypara}[1]{\smallskip\noindent\textbf{#1.}}  % Paragraph headers
\makeatletter
\newtheorem*{rep@theorem}{\rep@title}
\newcommand{\newreptheorem}[2]{
\newenvironment{rep#1}[1]{
 \def\rep@title{#2 \ref{##1}}
 \begin{rep@theorem}}
 {\end{rep@theorem}}}
\makeatother

\newreptheorem{proposition}{Proposition}

\newtheorem{lemma}{Lemma}
\newreptheorem{lemma}{Lemma}

\newtheorem{corollary}{Corollary}
\newreptheorem{corollary}{Corollary}

\newtheorem{theorem}{Theorem}
\newreptheorem{theorem}{Theorem}

\newtheorem{claim}{Claim}

\newtheorem{fact}{Fact}

\newcommand{\capacity}{\textsf{Capacity}}
\newcommand{\wcapacity}{\textsf{WCapacity}}
\newcommand{\scheduling}{\textsf{Scheduling}}
\newcommand{\wscheduling}{\textsf{WScheduling}}
\newcommand{\uscheduling}{\textsf{UScheduling}}

\newcommand{\cG}{{\cal G}}
\newcommand{\cS}{{\cal S}}
\newcommand{\cI}{{\cal I}}

\begin{document}

\title{The Price of Local Power Control in Wireless Scheduling}

\author{
  Magn\'us M. Halld\'orsson
  \qquad
  Tigran Tonoyan \\ \\
  ICE-TCS, School of Computer Science \\
  Reykjavik University \\
  \url{{magnusmh,ttonoyan}@gmail.com}
}

\begin{titlepage}

\maketitle              % typeset the title of the contribution

\begin{abstract}
  We consider the problem of scheduling wireless links in the physical model, where we seek an assignment of power
  levels and a partition of the given set of links into the minimum
  number of subsets satisfying the 
  signal-to-interference-and-noise-ratio (SINR) constraints. Specifically, we are interested in the efficiency of local power assignment schemes, or \emph{oblivious power schemes}, in approximating wireless scheduling. Oblivious power schemes are motivated by networking scenarios when power levels must be decided in advance, and not as part of the scheduling computation.

We first show that the known algorithms fail to obtain sub-logarithmic bounds; that is, their approximation ratio are $\tilde\Omega(\log \max(\Delta,n))$, where $n$ is the number of links, $\Delta$ is the ratio of the maximum and minimum link  lengths,
and $\tilde\Omega$ hides doubly-logarithmic factors.
We then present the first $O(\log{\log\Delta})$-approximation algorithm,
which is known to be best possible (in terms of $\Delta$) for oblivious power schemes.
We achieve this by representing interference by a conflict graph, which allows the application of graph-theoretic results for a variety of related problems, including the weighted capacity problem. 
We explore further the contours of approximability, and find the choice of power assignment matters; that not all metric spaces are equal; and that the presence of weak links makes the problem harder.
Combined, our result resolve the \emph{price of oblivious power} for wireless scheduling, or the value of allowing unfettered power control.
\end{abstract}

\thispagestyle{empty}

\end{titlepage}

\section{Introduction}

% MAC layer scheduling
The task of the MAC layer in TDMA-based (time-division multiple access) wireless networks is to determine which nodes can communicate in which time-frequency slot.
A scheduler aims to optimize criteria involving throughput and fairness.
This requires obtaining effective spatial reuse while satisfying the interference constraints.
We treat the fundamental \emph{scheduling} problem of partitioning a given set of communication links into the fewest possible feasible sets with respect to interference constraints, as well as the related continuous problem.

% Conflict graphs
Abstracting wireless interference by \emph{conflict graphs} is a common practice in wireless research. However, arbitrary graphs are too general to be useful (in the worst case) for scheduling problems; hence, the standard modus operandi is to assume geometric intersection graphs, such as \emph{unit disk graphs} or the \emph{protocol model} \cite{kumar00}.
Unfortunately, disk graphs provably lack fidelity to the reality of wireless signals, being simultaneously too conservative and too loose \cite{moscitopology,moscibeyondgraphs}.
Yet, the graph abstraction has its advantages such as local and simple representation and it connects better to the literature.
We adopt the more accurate but complicated \emph{SINR model}, 
where signal decays as it travels and a transmission is successful if 
its strength at the receiver exceeds the accumulated signal strength of interfering transmission by a sufficient (technology determined) factor.
Even here, the standard analytic assumption that signal decays polynomially with the distance traveled is far from realistic \cite{son2006,MaheshwariJD08}, 
but it has been shown that results obtained with that assumption can be translated to the setting of arbitrary measured signal decay \cite{us:PODC14,us:MSWiM14} and statistical models~\cite{rayleigh}.

We try to combine ideas from the research in both models above to treat the scheduling problem.

% The three fundamental problems, formal definitions
% -- Including using power control
\mypara{Problem formulations}
Given as input is a set $\Gamma$ of $n$ communication \emph{links}; each link is a 
pair of a sender and receiver nodes in a metric space.
The senders can adjust their transmission power as needed. %; in each of the problems we consider, finding the appropriate power assignment is a part of the problem. 
A subset $S \subseteq \Gamma$ of links is \emph{feasible} if there exists a power assignment for which the transmission on each link satisfies the SINR formula (see Section 2) when the links in $S$ transmit simultaneously.
% Problem statements
We treat the following two problems:
%\begin{quote}

\scheduling: Partition $\Gamma$ into fewest number of feasible sets.

\wcapacity: Find the maximum weight feasible subset of $\Gamma$, when the links have
positive weights.

% Importance of weighted capacity
The {\wcapacity} problem is of fundamental importance to dynamic scheduling where requests
arrive over time. In a celebrated generic result, Tassiulas and Ephremides \cite{TE92}
show that an optimal scheduling strategy is to schedule in each slot
a maximum set of links weighted by the number of packets waiting.
Approximating {\wcapacity} thus results in dynamic packet scheduling with 
equivalent throughput approximation.

When {\wcapacity} is restricted to unit weights, we get the related unweighted {\capacity} problem.

Power control is a crucial component of wireless scheduling that may dramatically affect scheduling efficiency. Optimal solutions may require context-sensitive power assignments, where the power assigned to a
link depends on all the other links. However, in many scenarios of the multifaceted area of wireless network optimization the links may be bound to use only local information (together with a common strategy) when adjusting their power levels. 
In such power regimes, which are called \emph{oblivious power schemes}, the power chosen for a link depends only on the link itself, 
specifically on the link length. The main question that we address is how well can {\scheduling} and {\wcapacity} be approximated using only oblivious power schemes and whether such approximation can be obtained efficiently.
Note that we still compare 
to the solutions with optimal power control.

\mypara{Related Work}
% Models
Gupta and Kumar \cite{kumar00} proposed the geometric version of SINR and initiated average-case analysis 
of capacity known as scaling laws. Considerable progress has been made in recent years in elucidating essential algorithmic properties of the SINR model (e.g., 
\cite{moscibrodaconnectivity,GoussevskaiaHW14,AvinEKLPR12,dinitz,KV10,KesselheimSODA11,DaumGKN13,JurdzinskiKRS14}).
%\cite{moscibrodaconnectivity,GoussevskaiaHW14,AvinEKLPR12,dinitz,KV10,KesselheimSODA11,SODA12,DaumGKN13,JurdzinskiKRS14}).
 Moscibroda and Wattenhofer \cite{moscibrodaconnectivity} initiated
worst-case analysis in the SINR model. 
%The majority of analytic wireless research has actually been in various disk graphs.
%An alternative approach is to 
%
% Sched: early work + NP-hardness
Early work on the {\scheduling} problem includes \cite{elbatt02,CruzS03,chafekar07,moscihowoptimal}.
NP-completeness has been shown for {\scheduling} with different forms of power control: none
\cite{goussevskaiacomplexity}, limited \cite{katz2010energy}, and unbounded \cite{lin2012complexity}. 
Distributed algorithms attaining $O(\log n)$-approximation are also known \cite{KV10,icalp11}.

% Known results 
The unweighted {\capacity} problem has 
an efficient constant-factor approximation, due to Kesselheim \cite{KesselheimSODA11}
that holds for general metrics \cite{KesselheimESA12}, as well as for versions with
various fixed power assignments \cite{SODA11}.
This immediately yields $O(\log n)$-approximation for {\scheduling}. 
A different approach is to divide the links into groups of nearly equal length
and schedule each group separately. 
With this approach, numerous $O(\log \Delta)$-approximation results have been argued 
\cite{goussevskaiacomplexity,fu2009power,us:talg12}, where $\Delta$ is the ratio between the longest and shortest link length. In a recent work, we proposed a novel conflict-graph based approach that yields $O(\log^*{\Delta})$ approximation for {\scheduling} and {\wcapacity}~\cite{us:stoc15}. No constant factor approximation algorithm is known for these problems.

% Oblivious power assignments: previous results and our results
It was shown in \cite{FKRV09} that every oblivious power assignment can be worst possible factor $n$ from optimal. In terms of the parameter $\Delta$, however, the bound becomes $\Omega(\log\log\Delta)$~\cite{FKRV09,us:talg12}, i.e. the factor $n$ may appear only when the network contains doubly-exponentially long links. Indeed, for {\capacity}, there is an algorithm using oblivious power that achieves  $O(\log\log \Delta)$-approximation \cite{halwatpowerofnonuniform}, which gives $O(\log\log \Delta \log n)$-approximation for {\scheduling} (see also \cite{us:talg12}).

Thus, our aim is to narrow the gap between oblivious power {\scheduling} approximations involving the factors $\log n$ and $\log \Delta$ and the sub-logarithmic lower bound $\log \log \Delta$.

\mypara{Further Related Work}
{\scheduling} and {\wcapacity} have also been considered for \emph{fixed oblivious power assignments}~\cite{halwatapx, FKRV09, us:talg12, halmitcognitive,fangkeslinear}. The only known constant-factor approximation algorithms for these problems are obtained in the case of the \emph{linear power scheme}~\cite{halmitcognitive, tonoyanlinear}.
%
%The locality question of oblivious power schemes was also approached in~\cite{us:talg12, tonoyanmeancapacity}, again, from the perspective of a fixed power assignment.

% Our main results
\mypara{Our Results} First, we examine the previous approaches for {\scheduling} with oblivious power assignments that are known to provide  $O(\log{\Delta})$ (in terms of only $\Delta$) or $O(\log{n})$ (in terms of only $n$) approximation and show that these approximation guarantees cannot be improved. This motivates our main result, which is a $O(\log\log \Delta)$-approximation algorithm for {\scheduling} and {\wcapacity} using oblivious power assignments, which matches the known lower bounds \cite{FKRV09,us:talg12}.
Unlike the state of affairs for the {\capacity} problem, the results are surprisingly sensitive to the
metric and the exact power assignment.
They hold for doubling metrics, but provably fail in general (or even tree) metrics,
and they hold for all power assignments $P_\tau$ with $m/\alpha < \tau < 1$, where $m$ is the doubling dimension of the metric space,
while we show that using other values of $\tau$ fail.
%the use of other oblivious power assignments $P_\tau$ fails.
% Conflict graphs & our contributions
Our bounds are obtained by using the conflict graph-based framework introduced in~\cite{us:stoc15}. This entails finding, for a given set of links, a pair of conflict graphs $A$ and $B$, having the property that independent sets
in $A$ correspond to (SINR-) feasible sets of links, when using the
right (oblivious) power assignment and feasible sets correspond to independent sets in $B$, while the chromatic numbers of the two graphs are close to each other. Thus, at a low cost, we effectively
simplify the all-to-all SINR model by a pairwise relationship, in
fact a graph class for which the core problems are
constant-approximable. 
% A case in point, illustrating the benefit of the graph formulation,
%we find that known graph results immediately imply improved online algorithms for {\scheduling} and {\capacity}.
% Technical contributions
A key property allowing construction of graphs $A$ is locality, 
%something that has eluded all SINR work so far, 
which boils down to being safe from the 
effect of links that are far away. We somewhat surprisingly find that locality is achieved 
only for a special sub-family of oblivious power assignments.%, when the exponent of the power scheme is larger than 2 (in the 2-dimensional plane) yet smaller than $\alpha$, the path loss constant.

We also sketch a distributed algorithm for finding the $O(\log\log\Delta)$-approximate solution for {\scheduling}, which, however, needs further elaboration.

\mypara{Roadmap} Concepts and formal definitions are given in the next section. The limitations of several known approaches for {\scheduling} are discussed in Sec.~\ref{sec:knownalgos}. Sec.~\ref{sec:graphapprox} contains the main result: a $O(\log{\log{\Delta}})$ approximation algorithm for {\scheduling} and {\wcapacity} using oblivious power assignments. Limitation results on the use of general metrics or different oblivious power assignments are given in Sec.~\ref{sec:limitation}, and the notes on distributed computation of schedules are given in Sec.~\ref{sec:distributed}. Most proofs have been relegated to the appendix.

\section{Model and Definitions}\label{sec:model}

\mypara{Communication Links}
Consider a set $\Gamma$ of $n$\label{G:numlinks} \emph{links}, numbered from $1$ to $n$. Each link $i$ represents a unit-demand communication request from a sender $s_i$\label{G:siri} to a receiver $r_i$ - point-size wireless transmitter/receivers located in a metric space with distance function $d$\label{G:distance}. We denote $d_{ij}=d(s_i,r_j)$\label{G:asymdistance} the distance from the sender of link $i$ to the receiver of link $j$,  $l_i=d(s_i,r_i)$\label{G:li}  the \emph{length} of link $i$ and $d(i,j)=d(j,i)$\label{G:symdistance} the minimum distance between a node of link $i$ and a node of link $j$. 
 We let $\Delta(\Gamma)$\label{G:delta} denote the ratio between the longest and shortest link lengths in $\Gamma$, and drop $\Gamma$ when clear from context. We call a set of links $S$ \emph{equilength} if $\Delta(S)\le 2$.

\mypara{Power Schemes}
A \emph{power assignment} for $\Gamma$ is a function $P:\Gamma\rightarrow \mathbb{R}_+$\label{G:power}. For each link $i$, $P(i)$ defines the power level used by the sender node $s_i$.  
We will be particularly interested in power assignment schemes or \emph{power schemes} $P_{\tau}$\label{G:powertau} of the form $P_{\tau}(i)=C\cdot l_i^{\tau\alpha}$, where $C$ is constant for the given network instance. These are called \emph{oblivious power assignments} because the power level of each link depends only on a local information - the link length. Examples of such power schemes are  \emph{uniform} power scheme ($P_0$), \emph{linear} power scheme ($P_1$) and \emph{mean} power scheme ($P_{1/2}$) \cite{FKRV09}.

\mypara{SINR Feasibility}
In the \emph{physical model (or SINR model)} of communication~\cite{rappaport}, a transmission of a link $i$ is successful if and only if 
\begin{equation}\label{E:sinr}
\cS_i\ge \beta\cdot \left(\sum_{j\in S\setminus \{i\}}{\cI_{ji}} + N\right),
\end{equation}
where $\cS_i$ denotes the received signal of link $i$, $\cI_{ji}$ denotes the interference on link $i$ caused by link $j$, 
$N\ge 0$\label{G:noise} is a constant denoting the ambient noise, $\beta>1$\label{G:beta} is the minimum SINR (Signal to Interference and Noise Ratio) required for a message to be successfully received and $S$ is the set of links transmitting concurrently with link $i$. If $P$ is the power assignment used, then $\cS_i=\frac{P(i)}{l_{i}^\alpha}$ and $\cI_{ji}=\frac{P(j)}{d_{ji}^\alpha}$, where $\alpha\in (2,6)$\label{G:alpha} is the path-loss exponent.

A set $L$ of links is called $P$-\emph{feasible} if the condition~(\ref{E:sinr}) holds for each link $i\in L$ when using power assignment $P$. We say $L$ is \emph{feasible} if there exists a power assignment $P$ for which $L$ is $P$-feasible. Similarly, a collection of sets is $P$-feasible/feasible if each set in the collection is.

\mypara{Capacity and Scheduling Problems}
{\scheduling} denotes the problem of partitioning a given set $\Gamma$ into the minimum number of feasible subsets (or \emph{slots}).
{\wcapacity} denotes the problem where we are also given a weight function $\omega:\Gamma\rightarrow \mathbb{R}_+$ on the links and we seek a maximum weight feasible subset $S$ of $\Gamma$.

\mypara{Affectance}
Following~\cite{halwatapx}, we define  the \emph{affectance} $a_P(i,j)$\label{G:affectance} of link $i$ by link $j$ under power assignment $P$ by
\[
a_{P}(j,i)=c_i\frac{\cI_{ji}}{\cS_i}=c_i\frac{P(j)l_i^{\alpha}}{P(i)d_{ji}^{\alpha}},
\]
where $c_i=1/(1-\beta N l_i^\alpha/P(i))$ is a factor depending on the properties of link $i$\footnote{If the denominator of $c_i$ is $0$, i.e. $P(i)=\beta N l_i^\alpha$, then link $i$ must always be scheduled separately from all other links. We assume that there are no such links.}. We let $a_{P}(j,j)=0$ and extend $a_{P}$ additively over sets: $a_P(S, i)=\sum_{j\in S}{a_{P}(j,i)}$ and $a_P(i, S)=\sum_{j\in S}{a_{P}(i, j)}$. It is readily verified that a set of links $S$ is feasible if and only if $a_{P}(S,i) \le 1/\beta$ for all $i\in S$. We call a set of links \emph{$p$-$P$-feasible} for a parameter $p>0$ if $a_{P}(S,i) \le 1/p$.

The following theorem exhibits the flexibility of the SINR threshold value, which has proved useful in obtaining asymptotic results.
 \begin{theorem}\cite{HB14}\label{T:signalstrengthening}
Any $p$-$P$-feasible set can be partitioned into $\left\lceil 2p'/p\right\rceil$ subsets, each of which is $p'$-$P$-feasible.
 \end{theorem}

We make the standard assumption that for all links $i$ in the instance,
received signal strength is a little larger than necessary to overcome the noise term $N$ alone in the absence of any
other transmissions: $P(i) \ge c \beta N l_i^\alpha$ for some constant $c>1$. This can be achieved by scaling the power levels of links or not having links that are too long.
% (see Sec.~\ref{sec:weaklinks}). 
This assumption helps to avoid the terms $c_i$ in the affectance formula. Indeed, it implies that $c_i\le c/(c-1)$ for all $i$. Then given e.g. a {\scheduling} instance $\Gamma$, we can solve it with $c_i=1$ for all $i$ and $\beta'=(c-1)\beta/c$, getting a feasible solution for the original problem. Moreover, by Thm.~\ref{T:signalstrengthening}, the number of slots obtained will be at most a constant factor away from the optimum of the original problem. Thus, we assume henceforth that $c_i=1$ for all links $i$, i.e. $a_{P}(i,j)=\frac{P(j)l_i^{\alpha}}{P(i)d_{ji}^{\alpha}}$. We have, in particular, $a_{P_{\tau}}(i,j)=\frac{l_i^{(1-\tau)\alpha}l_j^{\tau\alpha}}{d_{ij}^{\alpha}}$.

\emph{Remark.} In practice, there is an upper limit $P_{max}$ on the available power level of links and for some links, even setting $P(i)=P_{max}$ can be insufficient for having $P(i)\ge c\beta N l_i^\alpha$. Such links are called \emph{weak links}. Our assumption thus amounts to excluding weak links. Weak links are further discussed in 
%Sec.~\ref{sec:weaklinks}.
Sec.~\ref{sec:limitation}

\mypara{Fading Metrics}
The \emph{doubling dimension} of a metric space is the infimum of all numbers $\delta > 0$ such that every ball of radius $r>0$ has at most $C\epsilon^{-\delta}$ points of mutual distance at least $\epsilon r$ where $C\geq 1$ is an absolute constant, $\delta >0$ and $0<\epsilon \leq 1$. 
Metrics with finite doubling dimensions are said to be \emph{doubling}. For example, the $m$-dimensional Euclidean space has doubling dimension $m$~\cite{heinonen}.\label{G:dimension}
We will assume for the rest of the paper that the links are located in a doubling metric space with doubling dimension $m < \alpha$. Such metrics are called \emph{fading metrics}.

\begin{table}
\centering
    \begin{tabular}{  l  l c l }
		\textit{Notation} & \textit{Meaning} & \textit{Topic} & \textit{Page}\\\hline
    %$f^{(c)}(x)$ & function $f$ applied $c$ times & \textit{Functions} & {\pageref{G:frepeated}}\\
    %$f^*(x)$ & iterated $f$ &  & \pageref{G:fstar}\\\hline
    %
		$m$ & the doubling dimension of the metric space & \textit{Metric Space} & \pageref{G:dimension}\\
    $d$ & the distance function of the metric space & & \pageref{G:distance}\\\hline
		$n$ & the number of links & & \pageref{G:numlinks} \\ 
    $s_i,r_i$ & sender and receiver nodes of link $i$ & & \pageref{G:siri} \\
    $l_i$ & the length of link $i$, $l_i=d(s_i,r_i)$ & \textit{Links} & \pageref{G:li}\\ 
    $d_{ij}$ & the distance from $s_i$ to $r_j$, $d_{ij}=d(s_i,r_j)$ & & \pageref{G:asymdistance} \\
    $d(i,j)$ & the minimum distance between links $i,j$ &  & \pageref{G:symdistance}\\
    $\Delta(L)$ & the maximum ratio between link lengths in $L$ & & \pageref{G:delta} \\\hline
    $\chi(G)$ & the chromatic number of graph $G$ & & \pageref{G:chi}\\
    $\cG_f(L)$ & the $f$-conflict graph over the set $L$ & {\textit{Graphs}} & \pageref{G:gf} \\
    $\cG_\gamma(L)$ & the $f$-conflict graph over $L$ with $f(x)\equiv\gamma$ & & \pageref{G:ggamma}\\
    $\cG_\gamma^\delta(L)$ & the $f$-conflict graph over $L$ with $f(x)=\gamma x^\delta$ & & \pageref{G:ggammadelta}\\\hline
    $P$ & power assignment, $P:L\rightarrow \mathbb{R}_+$ & & \pageref{G:power}\\
    $P_{\tau}$ & oblivious power scheme given by $P_{\tau}(i)\sim l_i^{\tau\alpha}$ & & \pageref{G:powertau}\\
    $\alpha$ & the path loss exponent & \textit{SINR} & \pageref{G:alpha}\\
    $\beta$ & the SINR threshold value & & \pageref{G:beta}\\
    $N$ & the ambient noise term &  & \pageref{G:noise}\\
%    $OPTS(L)$ & the optimum schedule length of set $L$ & & \pageref{G:opts}\\
    $a_{P}(i,j)$ & $=c_j\frac{P(i)l_j^\alpha}{P(j)d_{ij}^\alpha}$, affectance of link $j$ by link $i$ & & \pageref{G:affectance}\\
    \end{tabular}
    \caption{Commonly used notations.}
\end{table}

\section{Limitations of Known Approaches}\label{sec:knownalgos}

We start by considering two algorithms that have been proven to achieve $O(\log{n})$ approximation for {\scheduling} with fixed oblivious power schemes. We show that this approximation bound cannot be essentially improved for these algorithms. Moreover, we show that in terms of only the parameter $\Delta$, the approximation factor is not better than $O(\log{\Delta})$. To achieve this, we construct network instances on the real line for which these algorithms perform relatively poorly.

\mypara{The First-Fit Algorithm}
The \emph{first-fit} algorithm considered in~\cite{halwatapx} was originally used for the uniform power scheme, but can be adapted (using the constant factor approximation algorithm for {\capacity} from~\cite{SODA11}) for other oblivious power schemes as well. The algorithm is a simple greedy procedure, where one starts with empty slots in a fixed order, then the links are processed in \emph{increasing order by length} and a link is assigned to the first slot that is feasible together with the link under consideration. It is known that the first-fit algorithm achieves an approximation factor of $O(\log{n})$ for a wide range of power assignments in general metric spaces~\cite{SODA11}.

The following family of hard network instances for the first-fit algorithm is inspired by a well known tree construction for online graph coloring~\cite{GyarfasL88} and its geometric realization as a disk graph~\cite{CaragiannisFKP07}. The rooted tree $T_k$, ($k\ge 0$), is constructed recursively, as follows. $T_0$ consists of a single root node. For $k\ge 0$, the tree $T_{k+1}$ is obtained from $T_k$ by adding a new child node to the root, then adding a copy of $T_k$, by identifying its root node with the new child. For example, $T_1$ consists of two nodes connected by an edge and $T_2$ consists of a root node that has two children and one ``grandchild''. 
Note that the number of nodes in $T_k$ is $n=2^k$ and the depth is $k=\log{n}$. Let us call the set of leafs of $T_k$ \emph{layer $k$}. For $t=k-1,k-2,\dots,1$, \emph{layer $t$} denotes the set of leafs of the tree that remains after removing layers $k,k-1,...,t+1$. Thus, $T_k$ has $k+1$ layers and the root is in layer $0$. Note also that each layer $t$-node has exactly one child from each of layers $k,k-1,\dots,t+1$.

We construct a set $L_k$ of links on the real line, where each link corresponds to a vertex of $T_k$. Let us fix a value $\delta\in [0,1)$. Links corresponding to the same layer in the tree have equal length, which decreases with increasing layer numbers, and the lengths are such that $\Delta(L_k)=c^k$ for a constant $c>1$. If two links correspond to adjacent vertices they cannot be in the same $P_{\delta}$-feasible slot, otherwise they are spatially well separated. By the results of Sec.~\ref{sec:graphapprox}, $L_k$ can be scheduled in a constant number of slots using an oblivious power scheme $P_{\tau}$. On the other hand, it follows from the construction that when $L_k$ is fed to the first-fit algorithm using $P_{\delta}$ in an increasing order by length, different ``layers'' of links get to be scheduled in different slots, thus giving only a $\Theta(k)=\Theta(\log{\Delta})=\Theta(\log{n})$ approximation.
\begin{theorem}\label{T:firstfit}
Let $\delta\in [0,1)$. For each $N>0$, there is a set of $n>N$ links $L$ on the real line s.t. any first-fit algorithm that treats the links in an increasing order of length and uses power scheme $P_{\delta}$ achieves no better than $\Omega(\log{\Delta})=\Omega(\log{n})$ approximation for {\scheduling}. Moreover, if $\delta > 1/\alpha$ then the approximation bound holds even for the {\scheduling} problem w.r.t. fixed power scheme $P_\delta$.
\end{theorem}
\begin{proof}
The following proof uses definitions and results from Sec.~\ref{sec:graphapprox}.
Let us fix a $k>0$ and $\delta\in [0,1)$ and let $x >0$ be a parameter to be defined below. Assume that $\beta=1$. We model the set $L_k$ of links after the tree $T_k$. Each link corresponds to a node of the tree. The links are arranged on the real line. The links corresponding to layer $t$ nodes have length $x^{k-t}$. For instance, the root has length $x^k$ and the leaves have length $1$. Note also that $\Delta(L_k)=x^k$. The root is placed with its sender on the origin and the receiver at the coordinate $x^k$. Assume links $i$ and $j$ are so that the node corresponding to $i$ is the parent of the node of $j$ in the tree (hence, $l_i>l_j$). Moreover, assume that the placement of link $i$ has already been determined. Then we place the link $j$ so that $s_j=r_i+d_{ji}=r_i + l_i^{1-\delta}l_j^\delta$ and $r_j=s_j+l_j$. Such placement guarantees that any two links corresponding to adjacent tree nodes cannot be in the same $P_{\delta}$-feasible set (recall that $\beta=1$). In particular, such a pair of links is $(1, \delta)$-conflicting. 

It remains to specify the value of $x$, in order to complete the construction. We will define $x$ so as to keep links corresponding to non-adjacent tree nodes independent in the sense that they do not affect each other much. 

Let us start by computing the diameter $d(L_k)$ of $L_k$, i.e. the distance from the sender node of the root link to the rightmost receiver node of the set. Note that $d(L_k)$ is the diameter of the subset of the links corresponding to the longest branch of the root in $T_k$, which contains exactly one link of length $x^t$ for $t=0,1,\dots,k$. Thus,
\[
d(L_k)=\sum_{t=0}^k{x^t} + \sum_{t=1}^k{x^{(1-\delta)t}\cdot x^{\delta(t-1)}},
\]
where the first sum corresponds to the lengths of the links, and the second one corresponds to the distances between adjacent links. We have further,
\[
d(L_k)=\sum_{t=0}^k{x^{t}}+\sum_{t=1}^k{x^{t-\delta}}
=\frac{x^{k+1}-1}{x-1} + x^{-\delta}\left(\frac{x^{k+1}-1}{x-1}-1\right)
=\frac{x^{k+1}-1}{x-1}(x^{-\delta} + 1) - x^{-\delta}.
\]
implying that $x^k<d(L_k)<4x^k$ when $x\ge 2$. Consider a link $i$ of length $x^p$ and two of its children $j,k$ of length $x^t$ and $x^{t+1}$ respectively. We want link $k$ to appear to the right of the whole ``subtree'' of links rooted at link $j$; namely, $d(r_i,s_k) > d(r_i,s_j) + 2x\cdot d(L_t)$. Since $i$ is the parent of $j$ and $k$, we have, by definition, $d(r_i,s_k)=x^{p\cdot(1-\delta)}\cdot x^{(t+1)\cdot \delta}$ and $d(r_i,s_j)=x^{p\cdot(1-\delta)}\cdot x^{t\cdot \delta}$. Thus, due to the bound $d(L_t) <4x^t$, it suffices to have: 
$x^{(1-\delta)p+\delta(t+1)} > x^{(1-\delta)p + \delta t} + 8x^{t+1}$ or 
\[
x^{(1-\delta)p+\delta t}(x^{\delta} - 1) >8x^{t+1}.
\] 
Recall that $p\ge t+2$ as link $i$ is strictly longer than its children. Thus, the requirement above boils down to $x^{t-\delta + 2} > 16x^{t+1}$, and thus to $x>16^{\frac{1}{1-\delta}}$. We choose $x$ to be any constant satisfying $x>16^{\frac{1}{1-\delta}}$. Note that with such choice of $x$, the set $L_k$ has the following properties: if two links correspond to adjacent nodes in the tree, they are $(1,\delta)$-conflicting; otherwise, they are $(1,1)$-independent. In particular, the graph $\cG_{1}^1(L_k)$ is isomorphic to $T_k$. Thus, $L_k$ can be split into a constant number of $(\gamma,1)$-independent sets for any constant $\gamma > 0$, by Thm.~\ref{T:constantsens} (which holds even for the linear function on the Euclidean plane, as shown in~\cite{us:stoc15}). By lemmas~\ref{L:mainlemma1} and~\ref{L:mainlemma2}, if the constant $\gamma$ is large enough, each of these subsets is $P_{\tau}$-feasible for any $\tau \in (\frac{1}{\alpha},1)$ (note that the links are in a $1$-dimensional doubling space). Since each layer $t$ link conflicts with a link from \emph{each} layer below, it will take $\Omega(k)$ slots to schedule $L_k$ using a first-fit algorithm with power scheme $P_{\delta}$. These observations imply both claims of the theorem.
\end{proof}

\mypara{The Randomized Algorithm}
 Next we consider the distributed algorithm (and its variants) presented in~\cite{KV10}. In this algorithm, the sender nodes of the links act in synchronous rounds and each sender node transmits with probability $p_i$ or waits with probability  $1-p_i$ in round $i$, where $p_i$ is the same for all links (but may change across the rounds). Once the transmission succeeds in round $i$, the node is silent in subsequent rounds. It is known that a certain choice of the probabilities guarantees an $O(\log{n})$ approximation (w.h.p) for the fixed power {\scheduling} problem  with many oblivious power assignments~\cite{icalp11,KV10}. Note that a family of network instances has been presented in~\cite{icalp11} for which the output of the algorithm is an $\Theta(\log{n})$-approximation, but this construction does not exclude that the $\Theta(\log{n})$ factor may be additive. In fact, the randomized algorithm schedules links in $O(opt + \log^2{n})$ slots when the linear power scheme is used ($\tau=1$), where $opt$ is the optimum schedule length w.r.t. $P_1$. As we show below, this is not the case for power schemes $P_{\tau}$ with $0 < \tau<1$.

Our construction in this case is also modeled after a tree. More precisely, we model a family of instances after the complete $\log^b{n}$-ary tree with $n$ nodes, where $b>0$ is a constant. %Here again, if two links correspond to adjacent nodes in the tree they are not $P_{\delta}$-feasible together and otherwise they are well separated. 
Let parameters $\delta$ and $M<n^{\epsilon}$ ($\epsilon\in (0,1)$ a constant) be fixed. We start by constructing a set of $n/M$ links on the line, where each link corresponds to a node of the complete $\log^{b}{n}$-ary tree of height $\Theta(\frac{\log{(n/M)}}{\log{\log{n}}})$ over a set of $n/M$ vertices. %The children of each link are contained in the interval occupied by that link.
%It is proved directly that the constructed set can be scheduled in a constant number of slots using the power scheme $P_{\delta}$. 
In order to complete the construction, we just replace each link with its $M$ identical copies, getting a set of $n$ links with $\Delta=n^{c}$ for a constant $c$ and optimum scheduling number $O(M)$. Using an analysis similar to~\cite[Thm. 6]{HalldorssonK14}, we prove the following theorem.

\begin{theorem}\label{T:randomized}
Let $\delta\in (0,1)$ and the probabilities $p_i (i=1,2,\dots)$ be fixed. For each $N>0$, there is a set of $n>N$ links $L$ on the real line s.t. the randomized algorithm that uses probabilities $p_i (i=1,2,\dots)$ yields only a $\Omega(\log{n}/\log{\log{n}})$ approximation to the {\scheduling} problem with fixed power scheme $P_\delta$, w.h.p. In terms of $\Delta$, the approximation factor is  $\Omega(\log{\Delta}/\log{\log{\Delta}})$.
\end{theorem}
\begin{proof}
 Let us assume, for simplicity, that $\beta=1$. We start with the description of a set of links simulating a rooted complete $\log^b{n}$-ary tree over a set of $n/M$ nodes, where $b>1$ is a constant to be chosen and $M=O(n^{\epsilon})$  ($\epsilon\in (0,1)$ a constant) is a parameter. We will often mix the terminology of links and trees, e.g. by speaking of children of links, hoping that will not cause any confusion. We split the tree into levels, where the root is at level $0$ and the nodes of (tree-) distance $t$ from the root constitute the level $t$. Note that the number of nodes at level $t$ is $\log^{tb}{n}$; hence, the number of levels is $k=\Theta\left(\frac{\log(n/M)}{\log{\log{n}}}\right)=\Theta\left(\frac{\log{n}}{\log{\log{n}}}\right)$.  For each $t\ge 0$, the level-$t$ links have equal length $\ell_t$. We assume that 
\begin{equation}\label{E:distrlengthratio}
\ell_t = c\ell_{t+1}\log^{d(t+1)}{n}
\end{equation}
 for large enough constants $c,d>0$. We describe the placement of links on the real line level by level, starting from level $0$, which contains a single link $i$. We set $s_i=0$, $r_i=s_i + \ell_0$, as shown in Figure~\ref{fig:distr}.
The children of link $i$ have length $\ell_1$. We place the $\log^b{n}$ child links of length $\ell_1$ inside the interval occupied by the link $i$, so that  (see Figure~\ref{fig:distr}):
\begin{enumerate}
  \setlength{\itemsep}{0cm}%
  \setlength{\parskip}{0cm}%
\item{the minimum distance of any two child links is at least $2\ell_1$,}
\item{the distance from any child to $r_i$ is at least $\ell_0^{1-\delta}\ell_1^\delta/2$,}
\item{the distance from any child to $r_i$ is at most $\ell_0^{1-\delta}\ell_1^\delta$,}
\item{the distance from $s_i$ to any child is at least $\ell_0/2$.}
\end{enumerate}
\begin{figure}[htbp]
\begin{center}
\includegraphics[width=0.44\textwidth, natwidth=1030, natheight=384]{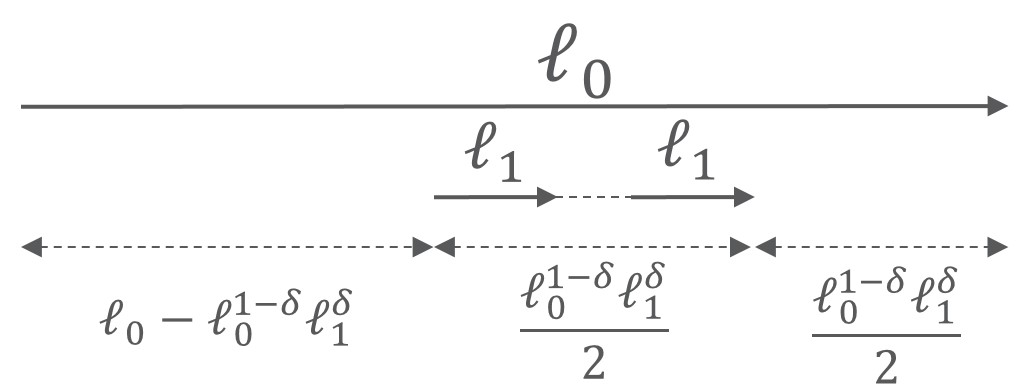}
\caption{The first step of the construction of Thm.~\ref{T:randomized}.}
\label{fig:distr}
\end{center}
\end{figure}
We set the constraints so as to have the following properties: 1. the set of links at the same level is (almost) feasible, 2. the children affect the parent, 3. the grand-children do not affect their grandparent, 4. the parent does not affect the children. The first three constraints will hold if $\ell_0^{1-\delta}\ell_1^\delta/2 > 3\ell_1\log^b{n}$, which holds if $\delta < 1$ and the constants $c,d$ in (\ref{E:distrlengthratio}) are large enough. The fourth constraint requires: $\ell_0-\ell_0^{1-\delta}\ell_1^\delta > \ell_0/2$, which holds if $\ell_0 > 2^{1/\delta}\ell_1$. This completes the first step of the construction. At the second step we construct the children of level-$1$ links in a similar fashion, and continue this process until having $n/M$ links. The length ratios defined by (\ref{E:distrlengthratio}) will ensure that the construction is correct and, in particular, that no link can be in the same feasible set as any of its children. On the other hand, we prove that the affectance on any level-$t$ link by all other links, except level-$t-1$ and level-$t+1$ links, is bounded by a constant. This implies that the set of links constructed can be scheduled in a constant number of slots using the power scheme $P_{\delta}$. Let $L$ denote the set of links and $L_t$ denote the set of level-$t$ links for $t\ge 0$.
\begin{claim}
If the constants $c,d$ in (\ref{E:distrlengthratio}) are large enough, then for any level-$t$ link $i$ ($t\ge 0$), it holds that $a_{P_{\delta}}(S,i)=O(1)$, where $S=L\setminus(L_{t-1}\cup L_{t+1})$.
\end{claim}
\begin{proof}
First, note that $a_{P_{\delta}}(L_t, i)= O(1)$ because all the links in $L_t$ have equal lengths and are well separated from each other. 
Now, let us fix an $s>t+1$. The number of level-$s$ links is $|L_{s}|=\log^{sb}{n}$. The distance from each level-$s$ link to $r_i$ is at least $\ell_{t}^{1-\delta}\ell_{t+1}^{\delta}/2$, by construction. Thus, we have:
\[
a_{P_{\delta}}(L_s, i)\le |L_s|\frac{\ell_s^{\delta\alpha}\ell_t^{(1-\delta)\alpha}}{(\ell_{t}^{1-\delta}\ell_{t+1}^{\delta}/2)^\alpha}=2^\alpha|L_s|\left(\frac{\ell_{s}}{\ell_{t+1}}\right)^{\delta\alpha}.
\]
Since the number of levels is  $O(\log{n})$, it is enough to have $a(L_s, i) < \frac{1}{\log{n}}$, which is provided if
\[
\ell_{t+1}>2^{1/\delta}(|L_s|\log{n})^{1/(\delta\alpha)}\ell_s=2^{1/\delta} \log^{(sb+1)/(\alpha\delta)}{n}\ell_s. 
\]
The last inequality holds if we set $d\ge 2b/(\alpha\tau)$ and $c\ge 2^{1/\delta}$ in (\ref{E:distrlengthratio}).
Now, let us consider a layer $s < t-1$ for $t>0$. Recall that the distance from each link of $L_s$ to $r_i$ is at least $\ell_s/2$, by construction. The affectance of $L_s$ can be bounded as follows:
\[
a_{P_{\delta}}(L_s, i) \leq |L_s|\frac{\ell_s^{\delta\alpha}\ell_t^{(1-\delta)\alpha}}{(\ell_s/2)^{\alpha}} < 2^\alpha|L_s|\left(\frac{\ell_t}{\ell_s}\right)^{\delta\alpha}.
\]
Hence, we can easily get $a_{P_{\delta}}(L_s, i) < \frac{1}{\log{n}}$ by tuning the constants $c$ and $d$ in~(\ref{E:distrlengthratio}). This yields the claim.
\end{proof}
Thus, the set $L$ can be scheduled into a constant number of feasible slots, by considering e.g. the union of the odd-numbered levels and the union of the even-numbered levels separately. In order to complete the construction, we replace each link in $L$ with its $M$ identical copies. Let $L'$ denote this set of links. Note that $|L'|=n$ and the optimum scheduling number of $L'$ w.r.t. $P_{\delta}$ is $\Theta(M)$. 

It remains to prove that for any sequence $p_i$, $i=1,2\dots$ and $p_i\in (0,1)$, the randomized algorithm using the probabilities $p_i$ will schedule $L'$ in $\Omega(kM)=\Omega(M\log{n}/\log{\log{n}})$ slots. To that end, it will be more convenient to analyze the algorithm in terms of the conflict graph $G$ corresponding to $L'$, rather than the set of links itself. Note that the graph $G$ is constructed by replacing each vertex of a complete $\log^{b}{n}$-ary tree on $n/M$ vertices with a $M$-clique, where the cliques corresponding to two adjacent vertices form a $2M$-clique. Obviously, $\chi(G)=2M$. Level-$t$ vertices in $G$ are the vertices corresponding to level-$t$ vertices in the tree.  Let the probabilities $p_i$, $i=1,2,\dots$ be fixed. We consider the following variant of the algorithm with relaxed constraints on transmissions. In round $i$, each remaining vertex $v$ of $G$ selects itself with probability $p_i$ and is removed from the graph in this round if it selects itself and no neighbor is selected.  Let $T_t$ denote the first time step when the size of a level-$t$ $M$-clique is halved. Let $H_t$ denote the event that the size of the smallest level-$s$ clique is at least $(1-1/\log{n})M$ before iteration $T_{t+1}+1$, for any $s\le t$.
\begin{claim}
Consider $0\le t< k$. Suppose that $T_{t+1} < M\log{n}$. Then $\mathbb{P}[H_t] = 1-O(n^{-\frac{M}{130\log{n}}+1})$.
\end{claim}
\begin{proof}
Let us consider any fixed $s\le t$. Let $R_{v,i}$ denote the event that a level-$s$ vertex $v$ is removed in iteration $i\le T_{t+1}$. Then we have:
$
\mathbb{P}[R_{v,i}]=p_i(1-p_i)^{d(v)}\le p_ie^{-p_id(v)}\le 1/d(v)\le 2/(M\log^b{n}),
$
where the first inequality follows because $1-x\le e^{-x}$ for $x\in (0,1)$ (here, $x=p_i$), the second one follows because $e^x\ge x$ for all $x\ge 0$ (here, $x=p_id(v)$), and the third one follows because $d(v)\ge M\log^b{n}/2$ before $T_{i+1} +1$. By the union bound, we have: $\mathbb{P}[\cup_{i=1}^{T_{t+1}}{R_{v,i}}] \le 2/\log^{b-1}{n}$, because $T_{t+1}\le M\log{n}$. Let $K$ be a level-$s$ clique and let $K_{T_{t+1}}$ denote the set of nodes in $K$ that survived the first $T_{t+1}$ rounds and let $\bar{K}_{T_{t+1}}=K\setminus K_{T_{t+1}}$. Then, by the argument above, we have that $\mathbb{E}[|\bar{K}_{T_{t+1}}|]\le 2M/\log^{b-1}{n}$ and $\mu=\mathbb{E}[|K_{T_{t+1}}|] \ge (1-2/\log^{b-1}{n})M > 9M/10$ if $n$ is large enough. By a standard Chernoff bound with $\delta=1/6$,
\[
\mathbb{P}[|K_{T_{t+1}}| < 3M/4] \le \mathbb{P}[|K_{T_{t+1}}| < (1-\delta)\mu] < e^{-\delta^2\mu/3} <e^{-\frac{9M/10}{108}}=O(n^{-\frac{M}{130\log{n}}}).
\]
The claim now follows by the union bound, as there are at most $n$ cliques.
\end{proof}
Observe that given the event $H_t$, the difference between the times $T_{t+1}$ and $T_t$ is at least $M/4$ if $n$ is large enough. Indeed, $H_t$ implies that in round $T_{t+1}$, the size of each clique in levels $t, t-1, \dots,0$ is at least $3M/4$, and in order for a clique of size $3M/4$ to become less than $M/2$, at least $M/4$ rounds must pass. Thus, $\mathbb{P}[T_t - T_{t+1}\ge M/4]\ge \mathbb{P}[H_t]=1-O(n^{-\frac{M}{130\log{n}}+1})$ holds for each fixed $t$. By the union bound, the probability that the event ${T_{t+1} - T_t \ge M/4}$ is violated for at least one $t$ is at most $O(k\cdot n^{-\frac{M}{130\log{n}}+1})=O(n^{-\frac{M}{130\log{n}}+2})$. Thus, if $M>130c\log{n}$, then with probability $1-O(n^{2 - c})$, it will take at least $k\cdot M/4=\Omega(M\log{n}/\log{\log{n}})$ steps until all the vertices of the graph are removed. This completes the proof, also taking into account that $\log{\Delta}=\Theta(\log{n})$.
\end{proof}

\section{Approximations Based on Conflict Graphs}\label{sec:graphapprox}

The main result of this section is a $O(\log{\log{\Delta}})$-approximation algorithm for {\scheduling} and {\wcapacity} based on the conflict graph method introduced in~\cite{us:stoc15}. Conflict graphs are graphs defined over the set of links. Let us call a conflict graph $A(L)$ an \emph{upper bound graph} for a set $L$, if there is a power scheme $P_{\tau}$ such that each independent set in $A(L)$ is $P_{\tau}$-feasible. Similarly, we call a graph $B(L)$ a \emph{lower bound graph} for $L$ if each feasible set induces an independent set in $B(L)$. Note that the chromatic numbers of $A(L)$ and $B(L)$ give upper and lower bounds for {\scheduling} with oblivious power schemes. Moreover, if the vertex coloring problem for $A(L)$ can be efficiently approximated, then the upper bound is constructive. Now, our aim is to construct upper and lower bound graphs, such that the gap between their chromatic numbers is bounded. The less the gap, the better colorings of $A(L)$ approximate {\scheduling} with oblivious power. The outline of this section is as follows. First, we present a family of conflict graphs introduced in~\cite{us:stoc15} and point out a sub-family of lower bound graphs $\cG_{\gamma}(L)$. Next, we present a family of upper bound graphs $\cG_{\gamma}^\delta(L)$ and show that the gap between the chromatic numbers is $O(\log{\log{\Delta}})$. A combination of these results then yields our main result.

\mypara{Conflict Graphs} Let $f:\mathbb{R}_+\rightarrow \mathbb{R}_+$ be a positive function. 
Two links $i,j$ are said to be \emph{$f$-independent} if
  \[ \frac{d(i,j)}{l_{min}} > f\left(\frac{l_{max}}{l_{min}}\right), \]
where $l_{min}=\min\{l_i,l_j\},l_{max}=\max\{l_i,l_j\}$, and otherwise they are \emph{$f$-conflicting}. 
A set of links is $f$-independent if they are pairwise $f$-independent.

Given a set $L$ of links, $\cG_f(L)$\label{G:gf} denotes the graph with vertex set $L$, where two vertices $i,j\in L$ are adjacent if and only if they are $f$-conflicting.

We will be particularly interested in conflict graphs $\cG_f$ with $f(x)\equiv \gamma$ and $f(x)=\gamma\cdot x^{\delta}$ for constants $\gamma>0$ and $\delta\in (0,1)$. We will use the notation $\cG_{\gamma}$\label{G:ggamma} in the former case and the notation $\cG_{\gamma}^{\delta}$\label{G:ggammadelta} in the latter case. We will refer to independence (conflict) in $\cG_{\gamma}$ as $\gamma$-independence ($\gamma$-conflict, resp.) and to independence (conflict) in $\cG_{\gamma}^{\delta}$ as $(\gamma,\delta)$-independence ($(\gamma,\delta)$-conflict, resp.). Note that $\cG_{\gamma}$ is equivalent to $\cG_{\gamma}^0$.

It will be useful to note that two links $i,j$ with $l_i\ge l_j$ are $\gamma$-independent iff $d(i,j)>\gamma l_j$ and are $(\gamma, \delta)$-independent iff $d(i,j) > \gamma l_i^{\delta}l_j^{1-\delta}$.

There are several important properties of conflict graphs $\cG_{\gamma}^\delta$ that we will use (see~\cite{us:stoc15} for the proofs). These properties hold in metrics of constant doubling dimension. We list the properties together with brief explanations. The additional definitions are only needed to understand how the results of~\cite{us:stoc15} are adapted for our conflict graphs.

The first property is that constant factor changes of the parameter $\gamma$ affect the chromatic number of $\cG_{\gamma}^\delta$ by at most constant factors.
%A function $f$ is called \emph{strongly sub-linear} if for each $c>0$, there is $c'>0$ s.t. $cf(x)/x < f(y)/y$ for all $x,y\ge 1$ with $x>c'y$. It is easy to see that $f(x)=\gamma x^\delta$ is strongly sub-linear for any $\delta\in (0,1)$.
Let $\chi(G)$ denote the chromatic number of a graph $G$\label{G:chi}.
\begin{theorem}\label{T:constantsens}
For any set $L$ and constants $\gamma,\gamma'>0$ and $\delta\in (0,1)$, $\chi(\cG_{\gamma}^\delta(L))=\Theta(\chi(\cG_{\gamma'}^\delta(L)))$.
\end{theorem}
A \emph{$k$-simplicial elimination order} of graph $G$ is an arrangement of the vertices from left to right  where for each vertex, the set of neighbors appearing to its right can be covered with k cliques. A graph is $k$-simplicial if it has a $k$-simplicial elimination order. It is known that vertex coloring and  maximum weighted independent set problems are $k$-approximable in $k$-simplicial graphs~\cite{ackoglu, kammertholey, yeborodin}. The second property is: graphs $\cG_{f}$ with appropriate function $f$ are constant-simplicial~\cite{us:stoc15}.
\begin{theorem}\label{T:algorithms}
The vertex coloring and maximum weighted independent set problems are constant factor approximable in graphs $\cG_{\gamma}^\delta(L)$ with $\delta\in (0,1)$.
\end{theorem}
Let $f$ be a sub-linear function. For each integer $c\ge 1$, the function $f^{(c)}(x)$\label{G:frepeated} is defined recursively by: $f^{(1)}(x)=f(x)$ and $f^{(c)}(x)=f(f^{(c-1)}(x))$ for $c>1$. Let $x_0=\inf\{x\ge 1, f(x)<x\}+1$; such a point exists for an appropriate sub-linear function. The function \emph{iterated} $f$, denoted $f^*(x)$\label{G:fstar}, is defined by:
\[
f^*(x)=\begin{cases}
\min_c\{f^{(c)}(x)\le x_0\},&\text{ if }x>x_0,\\
1, &\text{ otherwise}.
\end{cases}
\]
For $f(x)=x^\delta$, $f^*(x)$ is the minimum number of times $f$ should be repeatedly applied on $x$ in order to get the value below $2$. Thus, in this case $f^*(x)=\lceil\log_{\delta}{\log{x}}\rceil$.

The third property is: the chromatic numbers of $\cG_{\gamma}$ and $\cG_{\gamma f}$ are at most a factor of $O(f^*(\Delta))$ apart, which gives the gap of $O(\log{\log{\Delta}})$ when $f(x)=x^\delta$.
\begin{theorem}\label{T:chromaticgap}
$\chi(\cG_{\gamma}^\delta(L))=\chi(\cG_{\gamma}(L)) \cdot O(\log{\log{\Delta}})$ for any constants $\gamma >0$ and $\delta\in (0,1)$.
\end{theorem}
The fourth property shows that for appropriate constant $\gamma>0$, $\cG_{\gamma}(L)$ is a lower bound graph, i.e. each feasible subset of $L$ is an independent set in $\cG_{\gamma}(L)$.
\begin{theorem}\label{T:lowerbound}
If $\beta>1$ then there is a constant $\gamma>0$ s.t. each feasible set is $\gamma$-independent, i.e. $\cG_{\gamma}(L)$ is a lower bound graph.
\end{theorem}

\mypara{Upper Bound Graphs}
Here we show that for appropriate values of $\delta$ and $\gamma$, graphs $\cG_{\gamma}^\delta(L)$ are upper bound graphs, i.e. each independent set in $\cG_{\gamma}^\delta(L)$ is feasible with the appropriate oblivious power assignment. This complements the conflict graph framework for approximating oblivious power feasibility described in the beginning of the section, leading to an $O(\log{\log{\Delta}})$-approximation.

The general proof idea is borrowed from~\cite{us:stoc15}. Namely, in order to bound the affectance of an independent set $S$ of links on a given link $i$, we first split $S$ into length classes (i.e. equilength subsets) and bound the affectance of each length class on $i$ separately (lemmas~\ref{L:mainlemma1} and~\ref{L:mainlemma2}). Then we combine the obtained bounds in a series that converges under the assumption that the links are in a fading metric (Thm.~\ref{T:obliviouspowers}). The affectance of each length class on link $i$ is bounded using the common ``concentric annuli'' argument (Lemma~\ref{L:garbage}), where the rough idea is to split the space into concentric annuli centered at an endpoint of $i$, bound the number of links in each annulus using link independence and the doubling property of the space, then use these bounds to bound the affectances of links from different annuli and combine them into a series that converges by the properties of the space. The difference from the proofs of~\cite{us:stoc15} is that here we have to deal with both long and short links, while in~\cite{us:stoc15} we had to consider only the influence of shorter links.

We will obtain a slightly stronger result than feasibility. Our results hold in terms of the function $f_{\tau}(i,j)$ with a parameter $\tau \in [0,1]$, where for any two links $i,j$,
$
 f_{\tau}(i,j)= \frac{l_i^{\tau\alpha}l_j^{(1-\tau)\alpha}}{d(i,j)^\alpha}
 $ (note that we have $d(i,j)$ instead of $d_{ij}$ in the denominator).
 Note that for any pair of links $i,j$,
$
a_{P_{\tau}}(i,j)\leq f_{\tau}(i,j).
$ The function $f_{\tau}(i,j)$ is extended additively to sets of links, similar to the function $a_{P}(i,j)$.

In the following core lemma, we show that the affectance of an independent equilength set of links $S$ (i.e. $\Delta(S)\le 2$) on a separated fixed link $i$ can be bounded by the ratio of the length $l_i$ and the minimum length in $S$. The proof is the ``concentric annuli'' argument described above.
 \begin{lemma}\label{L:garbage}
 Let $\delta, \tau \in (0,1)$ and $\gamma \ge 1$, let $S$ be an equilength set of $1$-independent links, and let $i$ be a link s.t. $i,j$ are $(\gamma,\delta)$-independent for all $j\in S$. 
Then,
 \[
 \displaystyle f_{\tau}(S,i)\in
 O\left(\gamma^{m-\alpha} \left(\frac{l_i}{\ell}\right)^{(1-\tau)\alpha - \delta(\alpha-m)}\cdot \min\left\{1,\frac{l_i}{\ell}\right\}^{-\delta} \right),
 \]
where $\ell$ denotes the shortest link length in $S$. 
 \end{lemma}

\begin{proof}
We will use the following two facts. 
 \begin{fact}\label{F:convex}
 Let $\alpha\geq 1$ and $r\geq 0$ be real numbers. Then $ \frac{1}{ r^\alpha} - \frac{1}{(r+1)^\alpha}\leq \frac{\alpha}{(r+1)^{\alpha+1}}.$
 \end{fact}

\begin{fact}\label{F:integral}
Let $g(x) = \displaystyle\frac{1}{(q + x)^{\gamma}}$, where $\gamma > 1$ and $q>0$. Then $\sum_{r=0}^\infty {g(r)}\in \displaystyle O\left(\frac{1}{q^{\gamma-1}} + \frac{1}{q^{\gamma}}\right).$
\end{fact}

 First, let us split $S$ into two subsets $S'$ and $S''$ such that $S'$ contains the links of $S$ that are closer to $r_i$ than to $s_i$,
i.e. $S'=\{j\in S: \min\{d(s_j,r_i), d(r_j,r_i)\} \leq \min\{d(s_j,s_i), d(r_j,s_i)\}\}$ and $S''=S\setminus S'$. Let us consider the set $S'$ first.

 For each link $j\in S'$, let $p_j$ denote the  endpoint of $j$ that is closest to node $r_i$. Denote $q=\gamma l_i^{\delta}\ell^{-\delta}$.
Consider the subsets $S_1,S_2,\dots$ of $S'$, where 
 $
 S_r=\{j\in S': d(j,i)=d(p_j,r_i)\leq q \ell+(r-1)\ell\}.
 $
 Note that $S_1$ is empty: $d(j,i)>\gamma l_i^{\delta}\ell^{1-\delta}=q\ell$ for all $j\in S'$ because $i,j$ are $(\gamma, \delta)$-independent, and so $S'=\cup_{r=2}^\infty{S_r}$. Let us fix an $r>1$. 
Consider any two links $j,k\in S_r$ s.t. $l_j\geq l_k$. We have that
 $d(p_j,p_k) \ge d(j,k) > \ell$ ($1$-independence) and that $d(p_j,r_i)\leq \gamma q \ell+(r-1)\ell$ for each $j\in S_r$ (by the definition of $S_r$), so using the doubling property of the metric space, we get the following bound: 
\begin{equation}\label{E:strs}
|S_r|=|\{p_j\}_{j\in S_r}|\leq C\cdot \left(\frac{q\ell+(r-1)\ell }{\ell}\right)^{m} = C \left(q+r-1\right)^{m}.
\end{equation}
Note also that $l_j \leq 2\ell$ and $d(i,j) \ge q\ell+(r-2)\ell$ for any link $j\in S_r\setminus S_{r-1}$ with $r>1$; hence, 
\begin{equation}\label{E:feqs}
f_{\tau} (j, i) = \frac{l_j^{\tau\alpha} l_i^{(1 - \tau)\alpha}}{d(i,j)^\alpha}
 \leq \ell^{(\tau - 1)\alpha}l_i^{(1-\tau)\alpha}\left(\frac{2\ell}{q\ell+(r-2)\ell}\right)^\alpha
 =\frac{2^\alpha \ell^{(\tau - 1)\alpha}l_i^{(1-\tau)\alpha} }{\left(q+r-2\right)^\alpha}.
\end{equation}
Recall that $S_{r-1}\subseteq S_r$ for all $r>1$, $S_1=\emptyset$ and $S'=\cup_{r=2}^\infty{S_r}$. Using~(\ref{E:feqs}), we have:
\begin{align}
{f_{\tau}(S',i)} & = \sum_{r\geq 2}{\sum_{j\in S_r\setminus S_{r-1}}{f_{\tau}(j,i)}} \nonumber \\
& \leq \sum_{r\geq 2}{\left(|S_r|-|S_{r-1}|\right)\frac{2^\alpha \ell^{(\tau - 1)\alpha}l_i^{(1-\tau)\alpha}}{\left(q+r-2\right)^\alpha}} \nonumber \\
& = 2^\alpha \left(\frac{l_i}{\ell}\right)^{(1-\tau)\alpha} \sum_{r\geq 2}{|S_r|\left( \frac{1}{\left(q+r-2\right)^\alpha} - \frac{1}{\left(q+r-1\right)^\alpha} \right)},
\label{eq:fjibnd}
\end{align}
where the last equality is just a rearrangement of the sum.
The sum can be bounded as follows:
\begin{align*}
\sum_{r\geq 2} {|S_r|\left( \frac{1}{\left(q+r-2\right)^\alpha} - \frac{1}{\left(q+r-1\right)^\alpha} \right)}
& \leq \sum_{r\geq 2}{|S_r|\frac{\alpha}{(q+r-1)^{\alpha+1}}} \\
& \leq \sum_{r\geq 2}{\frac{C \alpha (q+r-1)^m}{(q+r-1)^{\alpha+1}}} \\
& = O\left(\sum_{r\geq 2}{\frac{1}{(q+r-1)^{\alpha-m+1}}}\right)\\
&  = O\left(\frac{1}{q^{\alpha - m}} + \frac{1}{q^{\alpha - m + 1}}\right) \\
&  = O\left(\frac{1}{(\gamma l_i^\delta \ell^{-\delta})^{\alpha - m}} + \frac{1}{(\gamma l_i^\delta \ell^{-\delta})^{\alpha - m + 1}}\right) \\
&  = O\left(\gamma^{m- \alpha }\left(\frac{\ell}{l_i}\right)^{\delta(\alpha - m)}\left(1 + \left(\frac{\ell}{l_i}\right)^{\delta}\right)\right),
\end{align*}
where the first line follows from Fact~\ref{F:convex}, the second one follows from (\ref{E:strs}) and the fourth one follows from Fact~\ref{F:integral}. 
Combined with (\ref{eq:fjibnd}), this completes the proof for the set $S'$. 

The proof holds symmetrically for the set $S''$. Recall that $S''$ consists of the links of $S$ which are closer to the sender $s_i$ than to the receiver $r_i$. Now, we can re-define $p_j$ to denote the endpoint of link $j$ that is closer to $s_i$, for each $j\in S''$. The rest of the proof will be identical, by replacing $r_i$ with $s_i$ in the formulas. This is justified by the symmetry of $(\gamma, \delta)$-independence.
\end{proof}

In the following two lemmas we bound the affectance of a set $L$ of independent links on a fixed link $i$ that is sufficiently separated from $L$. The two cases when $L$ consists of links longer than $i$ and shorter than $i$ are treated separately because they impose different conditions on parameters $\delta$ and $\tau$. The idea of the proof is to split $L$ into length classes, bound the affectance of each length class using Lemma~\ref{L:garbage} and then combine the obtained bounds in a geometric series that will be upper bounded, provided that $\delta$ and $\tau$ satisfy the conditions of the lemmas. 
 \begin{lemma}\label{L:mainlemma1}
 Let $L$ be a $1$-independent set of links and $i$ be a link s.t. $l_i\geq l_j$ and $i,j$ are $(\gamma,\delta)$-independent for all $j\in L$.  Then for each $ \tau > 1- \delta (1-m/\alpha)$,
 $
 f_{\tau} (L, i)= O\left(\gamma^{m-\alpha}\right).
 $
 \end{lemma}

  \begin{proof}
Let us split $L$ into length classes $L_1, L_2,\dots$  with 
$
L_t=\{j\in L: 2^{t-1}\ell \leq l_j<2^t \ell\},
$
 where $\ell$ is the shortest link length in $L$. Let $\ell_t$ be the shortest link length in $L_t$.
Note that each $L_t$ is an equilength $1$-independent set of links which are $(\gamma,\delta)$-independent from link $i$. Thus, the conditions of Lemma~\ref{L:garbage} hold for each $L_t$. Since all links in $L_t$ are shorter than link $i$, we conclude that
\[
{f_{\tau}(L_t,i)} = O\left(\gamma^{m-\alpha}\left(\frac{\ell_t}{l_i}\right)^{\delta(\alpha-m) - (1-\tau)\alpha }\right).
\]
Recall that $L_t$ are length classes and $\ell_t\geq 2^{t-1}\ell$. That allows us to combine the bounds above into a geometric series:
\[
{f_{\tau}(L,i)}= \sum_{1}^{\infty}{f_{\tau}(L_t,i)}\le \frac{C\cdot \gamma^{m-\alpha}}{l_i^{(1-\tau)\alpha - \delta(\alpha-m)}}\sum_{t=0}^{\lceil\log{l_i/\ell}\rceil}{(2^{t}\ell)^{(1-\tau)\alpha - \delta(\alpha-m)}},
\]
where $C$ is a constant. The upper limit of the last sum is obtained by the fact that link $i$ is not shorter than the longest link in $L$. Recall that  $\tau > 1- \delta (1-m/\alpha)$; hence, $\delta(\alpha-m)  - (1-\tau)\alpha> 0$.
Thus, the last sum is the sum of a \emph{growing} geometric progression and is $O(l_i^{(1-\tau)\alpha - \delta(\alpha-m)})$, implying the lemma.
 \end{proof}

 \begin{lemma}\label{L:mainlemma2}
 Let $L$ be a $1$-independent set of links and $i$ be a link s. t. $l_i\leq l_j$ and $i,j$ are $(\gamma,\delta)$-independent for all $j\in L$.  Then for each $\tau < 1- (1-\delta)(\alpha - m + 1)/\alpha$,
 $
 f_{\tau} (L, i)= O\left(\gamma^{m-\alpha}\right).
 $
 \end{lemma}

\begin{proof}
Let us split $L$ into length classes $L_1, L_2,\dots$, where 
$
L_t=\{j\in L: 2^{t-1}l_i \leq l_j<2^t l_i\}.
$
Note that each $L_t$ is a equilength $1$-independent set of links that are $(\gamma, \delta)$-independent from link $i$. Let $\ell_t$ denote the shortest link length in $L_t$. Recall that $\ell_t \ge 2^{t-1}l_i$. Thus, Lemma~\ref{L:garbage} implies:
\[
{f_{\tau}(L_t,i)} = O\left(\frac{1}{\gamma^{\alpha-m}}\left(\frac{l_i}{\ell_t}\right)^{(1-\tau)\alpha - (1-\delta)(\alpha-m+1)}\right)=O\left(\frac{1}{\gamma^{\alpha-m}}\left(\frac{1}{2^{t-1}}\right)^{(1-\tau)\alpha - (1-\delta)(\alpha-m+1)}\right).
\]
Recall that $\tau < 1-(1-\delta)(\alpha- m + 1)/\alpha$, implying $\eta=(1-\tau)\alpha - (1-\delta)(\alpha-m+1) > 0$.
Thus, we have:
\[
{f_{\tau}(L,i)}= \sum_{1}^{\infty}{{f_{\tau}(L_t,i)}}\leq \gamma^{m-\alpha}\sum_{t=0}^{\lceil\log{l_i/\ell}\rceil}{\frac{1}{2^{\eta t}}} = O\left(\gamma^{m-\alpha}\right),
\]
where $C$ is a constant.
 \end{proof}

The main theorem of this sub-section follows by combining Lemmas~\ref{L:mainlemma1}~and~\ref{L:mainlemma2}: if parameters $\delta$ and $\tau$ satisfy the conditions of both lemmas simultaneously, then a set of $(\gamma,\delta)$-independent links will be $P_{\tau}$-feasible, provided the constant $\gamma$ is large enough, since $\alpha > m$.
\begin{theorem}\label{T:obliviouspowers}
If $\delta\in (\delta_0,1)$ and the constant $\gamma>1$ is large enough, the graphs $\cG_{\gamma}^{\delta}(L)$ are upper bound graphs for any set $L$, where $\delta_0=\frac{\alpha - m +1}{2(\alpha - m) +1}$. Namely, there exists $\tau\in (0,1)$ s.t. any $(\gamma,\delta)$-independent set is $P_{\tau}$-feasible. Moreover, we can choose $\tau=\delta$ whenever $\delta>\alpha/(2\alpha - m)$ and $m>1$.
\end{theorem}

\begin{proof}
We need to show that Lemmas~\ref{L:mainlemma1}~and~\ref{L:mainlemma2} hold simultaneously for the given $\delta$ and certain $\tau\in (0,1)$. Then we can adjust $\gamma$ in order to make $L$ feasible. The constraints of the mentioned lemmas on $\delta$ and $\tau$ are as follows:
\begin{equation}\label{E:taudel}
\tau>1-\delta\frac{\alpha - m}{\alpha}\mbox{ and } \tau < 1- (1-\delta)\frac{\alpha-m+1}{\alpha}.
\end{equation}
So it is enough to show that any $\delta \in (\delta_0, 1)$ is a solution for the following system of inequalities: 
\begin{equation}\label{E:taueq}
0<1-\delta\frac{\alpha - m}{\alpha}< 1- (1-\delta)\frac{\alpha-m+1}{\alpha}<1.
\end{equation}
The first and third inequalities hold whenever $\delta < 1$ and $\alpha > m$. The second inequality is equivalent to $\delta > \delta_0$.
The conditions for choosing $\tau=\delta$ follow by setting $\tau=\delta$ in~(\ref{E:taudel}).
\end{proof}

\mypara{Putting the Pieces Together}
All the components of the conflict graph framework are ready now, and a direct application of the technique described at the beginning of this section together with Thm.s~\ref{T:algorithms}-\ref{T:obliviouspowers} yield our main result.% -- an $O(\log{\log{\Delta}})$-approximation algorithm for {\scheduling} and {\wcapacity}, using oblivious power schemes. 

\begin{theorem}\label{T:schapprox}
There are $O(\log{\log{\Delta}})$-approximation algorithms for  {\scheduling} and {\wcapacity} using oblivious power schemes. The approximation is obtained by approximating vertex coloring or maximum weighted independent set problems in $\cG_{\gamma}^\delta(L)$ with appropriate constants $\gamma$ and $\delta$.
\end{theorem}
\begin{proof}
We present the proof for {\scheduling}. The argument for {\wcapacity} is similar and is omitted.
The algorithm is as follows. Given an input set $L$ in a fading metric,
choose values of parameters $\gamma$, $\delta$ and $\tau$ in accordance with Thm.~\ref{T:obliviouspowers}, construct the graph $G_{\gamma}^{\delta}(L)$ and find constant factor approximate coloring (by Thm.~\ref{T:algorithms}). Output the color classes obtained. The feasibility of the output follows from Thm.~\ref{T:obliviouspowers}. The approximation factor follows from the combination of Thm.s~\ref{T:lowerbound} and \ref{T:chromaticgap}.
\end{proof}

\section{Limitations of Oblivious Power Schemes}\label{sec:limitation}

\mypara{Euclidean Metrics}
We proved that using certain  oblivious power schemes, it is possible to approximate {\scheduling} and {\wcapacity} within approximation factor $O(\log{\log{\Delta}})$. As shown in Thm.~\ref{T:halldorssontalg} below, this bound is essentially best possible when using oblivious power assignments. 
The following was shown in greater generality in \cite{us:talg12}.

\begin{theorem}\cite{us:talg12}\label{T:halldorssontalg}
For every power scheme $P_\tau$ there is an infinite family of feasible sets $S$ arranged in a straight line such that any schedule of $S$ using $P_\tau$ requires $\Omega(\log{\log{\Delta}})$ slots.
\end{theorem}

Recall that we obtained our approximations only for oblivious power schemes $P_{\tau}$ with $\tau$ falling in a specific sub-interval of $(0,1)$. What happens with the other oblivious power schemes? Interestingly, as we show below,  oblivious power schemes $P_{\tau}$ with $\tau$ outside the range stipulated by Lemmas~\ref{L:mainlemma1}~and~\ref{L:mainlemma2} yield only $O(\log{\Delta})$-approximation for {\scheduling} and {\wcapacity}.

We consider a family of sets $L$ of $(1,1)$-independent links that are located in the Euclidean plane (hence, $m=2$). With separation $\delta=1$, the range of oblivious power schemes $P_{\tau}$ making $L$ (almost) feasible according to Lemmas~\ref{L:mainlemma1}~and~\ref{L:mainlemma2} is:
$
2/\alpha < \tau < 1.
$
In the following theorem we show that  no scheduling algorithm can achieve better than $O(\log{\Delta})$-approximation of {\scheduling}  for the set  $L$ using a scheme $P_{\tau}$ with $\tau < \frac{2}{\alpha}$. An equivalent lower bound applies to {\wcapacity}.
\begin{theorem}\label{T:hardinstance}
  For infinitely many $n$, there is a set of $n$ pairwise $(1, 1)$-independent links in the plane that requires $\Omega(\log n)$ slots when using $P_{\tau}$, $\tau < 2/\alpha$. In terms of $\Delta$, the number of slots required is $\Omega(\log \Delta)$.
\end{theorem}

 \begin{proof}
We assume that $\beta = 1$.
We inductively construct a weighted set of links $S_t = S_t(q)$ in the plane, given a parameter $q$.
We shall denote by $S_t^{(x,y)}$ a copy of the instance $S_t$ translated by the vector $(x,y)$. 

The instance $S_0$ consists of the single link $0$ of length $l_0=1$, with $s_0$ at the origin and $r_0$ at $(l_0,0) = (1,0)$. For $t \ge 1$, the instance $S_t$ consists of the link $t$ of length $3ql_{t-1}$ and of weight $\omega(t) = q^{2t}$ with $s_t$ at the origin and $r_t$ at  $(l_t,0)$, along with $q^2$ sub-instances $S_{t-1}^{(x_i,y_j)}$ with
$i,j=0, 1, \ldots q-1$, $x_i=2i l_{t-1}$ and $y_j =l_t + j (l_{t-1} + h_{t-1})$, where $h_t$ is the height of $S_{t-1}$.  This completes the construction.
See Figure \ref{fig:ht}.

\begin{figure}[htbp]
\begin{center}
\includegraphics[width=0.44\textwidth, natwidth=1075, natheight=795]{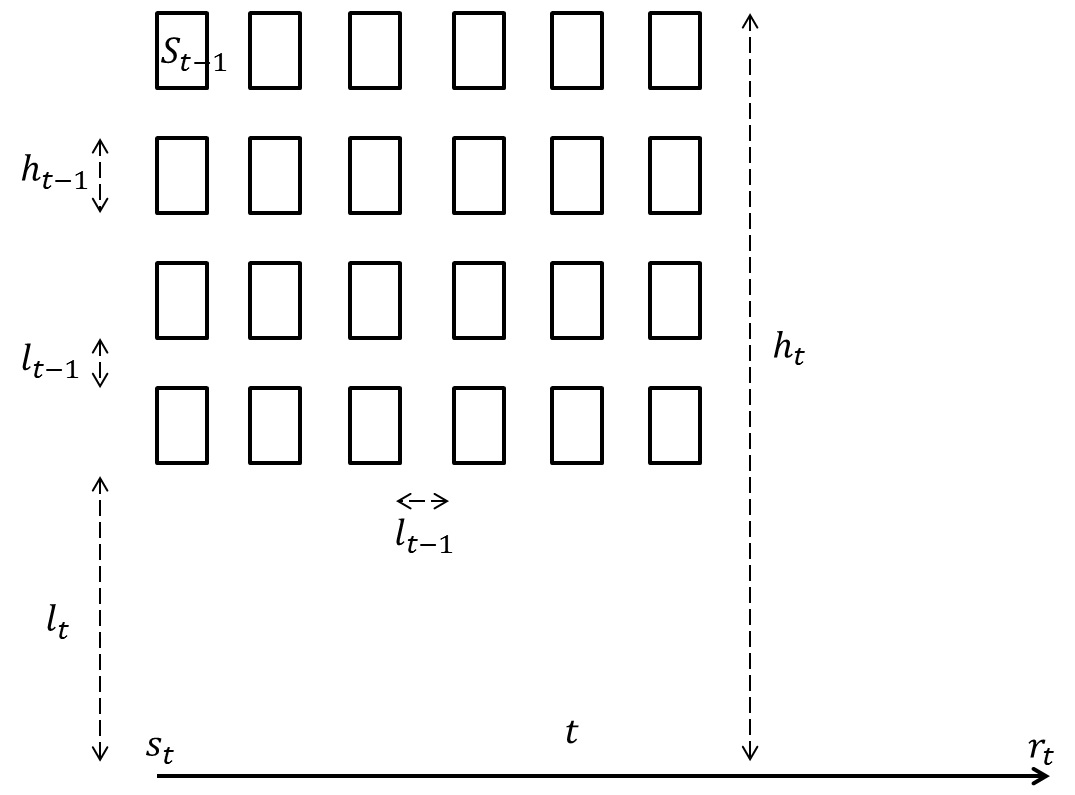}
\caption{$(1,1)$-separate instance $S_t$. Each rectangle represents a sub-instance isomorphic to $S_{t-1}$}
\label{fig:ht}
\end{center}
\end{figure}

It is easily verified that links in $S_t$ are $(1,1)$-separate; hence, $S_t$ can be scheduled in constant number of slots
using an oblivious power scheme, by Lemmas~\ref{L:mainlemma1},~\ref{L:mainlemma2} and Thm.~\ref{T:signalstrengthening}.
It remains to show that it requires $\Omega(\log n)$ slots when using $P_{\tau}$, with $\tau < 2/\alpha$.

Note that the number $n_t$ of links in $S_t$ is $n_t = 1 + q^2 n_{t-1} = \sum_{i=0}^{t-1} q^{2 i} = (q^{2t}-1)/(q^2-1)$. 
Thus, $\log n_t = \theta(t \log q)$. Let us call $t$ the \emph{main link} of $S_t$. Let us fix an index $t>0$. Let $L_k$ denote the set of main links of the copies of $S_k$ in $S_t$, where $k < t$. We call $L_k$ \emph{the $k$-th level of $S_t$}. All the links in $L_k$ have equal length and weight $q^{2k}$. It is easy to check that $W_t=\omega(L_k)=q^{2t}$,
and the total weight of all links is $t W_t = \theta(W_t \log n)$.

\begin{lemma}
Suppose $q \ge (2\cdot 3^{\tau\alpha})^{1/(2-\tau\alpha)}$.
Let $T$ be a subset of links in $S_t(q)$ that is feasible under $P_{\tau}$ with $\tau < 2/\alpha$.
Then, $\omega(T) \le 2\cdot 3^{\alpha}W_t$.
\label{lem:feas-wt}
\end{lemma}
\begin{proof}
First, an observation.
\begin{claim}
Let $T_k$ be a subset of level $k$ links in $S_t$ s.t. $T_k \cup \{t\}$ is $P_\tau$-feasible.
Then, $\omega(T_k) \le 3^{\alpha}2^{-(t-k)} W_t$.
\end{claim}
\begin{proof}
Let us first estimate the distance $d(i,t)$ for each link $i\in S_t\setminus\{t\}$. Note that $l_t=(3q)^t$, $h_t=l_t + q(l_{t-1} + h_{t-1})$ and $h_0=0$ (because $S_0$ consists of one horizontal link), so we can see that $h_t \leq 2l_t$. It follows that for each $i\in S_t\setminus\{t\}$,
$d(i,t)\leq 3l_t$, implying, for each $i\in T_k$,
 \[ 
 a_{P_{\tau}}(i, t)=\frac{P_{\tau}(i)l_t^\alpha}{P_{\tau}(t) d_{i,t}^\alpha}
   = \left(\frac{l_{i}}{l_{t}}\right)^{\tau\alpha} 
     \left(\frac{l_t}{d_{it}}\right)^\alpha
   \geq  \frac{1}{3^{\alpha}}(3q)^{-(t-i)\tau\alpha}. 
   \]
Since $a_{P_{\tau}}(L_k,t)\leq 1$, $T_k$ contains at most $3^{\alpha}(3q)^{(t-i)\tau\alpha}$ links, each of weight $q^{2k}$,
for a total weight of
\[
 \omega(T_k) \le 3^\alpha(3q)^{(t-k)\tau\alpha} \cdot q^{2k} 
  = 3^{\alpha}W_t \frac{(3q)^{(t-k)\tau\alpha}}{q^{2(t-k)}}
  = 3^{\alpha}W_t \left(\frac{(3q)^{\tau\alpha}}{q^{2}}\right)^{t-k}. 
  \]
The bound on $q$ ensures that $q^{2-\tau\alpha} \ge 2 \cdot 3^{\tau\alpha}$ or $q^2 \ge 2 \cdot (3q)^{\tau\alpha}$.
Thus, $\omega(T_k) \le 3^\alpha W_t (1/2)^{t-k}$, as claimed.
\end{proof}

We now prove the lemma by induction on $t$.
For $t=0$, $S_t$ consists of only one link of weight $1 = q^{0} = W_0$.
For the inductive step, we consider two cases. Suppose first that $T$ contains the link $t$.
Then, it follows from the claim that 
\[ 
\omega(T) \le \omega(t) + \sum_{k=0}^{t-1} \omega(T\cap L_k) 
    \le W_t + 3^{\alpha}W_t \sum_{k=0}^{t-1} 2^{-(t-k)} < 2\cdot 3^\alpha W_t.
\]
If, on the other hand, $T$ does not contain $t$,  it follows from the inductive hypothesis that the total weight of links from $T$ in each of the $q^2$ sub-instances $S_{t-1}^{x,y}$ is at most $2 \cdot 3^{\alpha}W_{t-1}$, for a grand total of 
$\omega(T) \le q^2 \cdot 2\cdot 3^{\alpha} W_{t-1} = 2\cdot 3^{\alpha} W_t$.
\end{proof}

Observe that the maximum length $\Delta(S_t) = \Delta_t$ of a link in $S_t$ is the length $l_t = (3q)^{t}$ of link $t$, which implies that $\log \Delta_t = \theta(t \log q) = \theta(\log n)$. Thus, $\Omega(\log \Delta)$ is also a lower bound.
\end{proof}

\begin{corollary}
Any algorithm for {\scheduling} that uses power assignment $
P_{\tau}$, $\tau < 2/\alpha$, is no better than $\Omega(\log{n})$ ($\Omega(\log{\Delta})$)-approximate in terms of $n$ (in terms of $\Delta$, resp.).
The same holds for {\wcapacity}.
\end{corollary}

As for the power schemes $P_\tau$ with $\tau\ge 1$, it is known that there is no algorithm using these power schemes that achieves better than $O(\log{\Delta})$-approximation in terms of $\Delta$. This is shown in~\cite{moscibrodaconnectivity} for the case $\tau=1$ (i.e. when the linear power scheme is used). It is also shown in~\cite{tonoyancapacity} that for any $P_{\tau}$ with $\tau>1$,  the optimal capacity is of the same order as the optimal capacity of $P_1$, which implies the claim for this case too.

\mypara{General Metrics}
\label{sec:sched}
Recall that for {\capacity}, the $O(\log\log \Delta)$-approximation results hold in arbitrary metrics~\cite{halwatpowerofnonuniform}.
This begs the question whether this might also hold for {\scheduling} and {\wcapacity}.
A negative answer was given for {\scheduling} in \cite[Thm.~5.1]{SODA11}: no bound of the form $f(\Delta)$, for any function $f$ of $\Delta$ alone. Namely, a feasible instance of $n$ \emph{equal length links} (i.e. $\Delta=1$) in a tree metric was given in \cite{SODA11}, for which $P_{\tau}$ (which is necessarily uniform power ($P_0$) on equal length links) requires $\Omega(\log n)$ slots.  Thus, there is a separation between possible bounds for {\capacity} and {\scheduling}. We simplify below this construction and show that it also gives the same lower bound for {\wcapacity}. 

\begin{theorem}\label{thm:wtcap-lb}
For infinitely many $n$, there is a feasible instance $L$ of  {\wcapacity} with $n$ equal length links in a metric space for which any solution that is feasible with uniform power is of weight only $O(\omega(L)/\log n)$.
\end{theorem}

\begin{proof}
We give a construction of a set of weighted equal length links that is feasible with a certain power assignment, but for which any 
subset that is feasible  using oblivious power, contains at most $\Omega(\log n)$ fraction of the total weight.  Since the links have equal lengths, the only possible oblivious assignment is the uniform one. This yields a $\Omega(\log n)$ lower bound on the price of oblivious
power for the weighted capacity problem.

\newcommand{\hi}{{1}}
\newcommand{\hj}{{\hat{j}}}

The set $L$ of links consists of $K$ subsets, $L_1, L_2, \ldots, L_K$ for $K>0$.   Each set $L_k$ contains $4^{k-1}$ links,
each of weight $1/|L_k|$, for a total weight of 1.  $L_k$ also has an associated number $t_k = (\gamma |L_k|)^{\frac{1}{\alpha}}$, for a constant parameter $\gamma$ to be determined. The distance between a link in $L_k$ and another link in $L_{k'}$ is simply $t_k + t_{k'}$.  We assume that $\beta = 1$. This completes the construction.  The total number of links $n = |L| =\sum_{k = 1 \ldots K} |L_k| = (4^K-1)/3$, and the total weight is $K$.

It was shown in \cite{SODA11} that $L$ is feasible using some power assignment. We give below a simplified proof.
We first show that any feasible set using uniform power has weight $O(1)$, or $O(1/\log n)$-fraction of the whole.

\begin{claim}
Let $S \subseteq L$ be a subset of links of weight $\omega(S) \ge 1 + \gamma 2^\alpha$. Then, $S$ is infeasible under uniform power.
\end{claim}
\begin{proof}
Let $\tilde{k}$ be the minimum value for which an element of $L_{\tilde{k}}$ exists in $S$. Consider an arbitrary link $l_j \in L_{\tilde{k}}\cap S$. Note that for $i \in L_k$ where $k > \tilde{k}$,  $d_{i j} = t_k + t_{\tilde{k}} \le 2 t_k = 2 (\gamma |L_k|)^{1/\alpha}$. The affectance $a(i,j) = a_{P_0}(i,j)$ under uniform power is then
\[
 a(i,j) = \frac{1}{d_{ij}^\alpha} \ge \frac{1}{\gamma 2^\alpha}
\cdot \frac{1}{|L_k|}\ . 
\]

Now,
\[
\sum_{i \in S} a({i},j) \geq \sum_{k > \tilde{k}} a({L_k \cap S},j)
  \ge \frac{1}{\gamma 2^\alpha} \sum_{k > \tilde{k}} \frac{|L_k \cap S|}{|L_k|}
  = \frac{\omega(S \setminus L_{\tilde{k}})}{\gamma 2^\alpha} 
  \ge \frac{\omega(S)-1}{\gamma 2^\alpha} > 1\ .
\]
\end{proof}

\begin{claim}
$L$ is feasible, assuming $\gamma \ge 6$.
\label{claim:feas}
\end{claim}

\begin{proof}
We will use the power assignment $P$ defined by $P(i) = \frac{1}{2^{k}}$, for $i \in L_k$.
Consider $j \in L_{\tilde{k}}$ and $i \in L_k$, for some $\tilde{k}$, $k$. 
Then, $d_{ij}^\alpha > t_{\max(k,\tilde{k})} = \gamma 2^{2(\max(k,\tilde{k})-1)}$. 
Thus,  
\[ 
a_P({L_k},j) = |L_k| \frac{2^{\tilde{k} - k}}{d_{ij}^\alpha}
  \le 2^{2(k-1)}  \cdot \frac{2^{\tilde{k} - k}}{\gamma \cdot 2^{2(\max(k,\tilde{k})-1)}} = \frac{1}{\gamma} 2^{\min(k,\tilde{k}) - \max(k,\tilde{k}) +1}\ . \]
It follows that
\[ a_{P}(L, j) = \sum_{k > \tilde{k}} a_P({L_k},j) + \sum_{k \le \tilde{k}} a_P({L_k},j)
  < \frac{1}{\gamma} \left(
    \sum_{k > \tilde{k}} 2^{\tilde{k}-k+1} + \sum_{k \le \tilde{k}} 2^{k - \tilde{k}+1} \right)
  < \frac{1}{\gamma} \left( 
    \sum_{x=0}^\infty \frac{1}{2^x} + \sum_{x=0}^\infty \frac{2}{2^x} \right)
  = \frac{6}{\gamma}\ . \]
Thus, for $\gamma \ge 6$, $L$ is feasible.
\end{proof}
Thm. \ref{thm:wtcap-lb} now follows.
\end{proof}

\begin{corollary}
  The use of oblivious power assignments cannot obtain approximations
  of {\scheduling} or {\wcapacity} within $o(\log n)$ factor in arbitrary general metrics.
\end{corollary}

\mypara{Weak Links}
%\section{Weak Links}
%\label{sec:weaklinks}
Recall that in order to obtain our approximations, we assumed that for each link $i$, $P(i)\ge c\beta N l_i^\alpha$ for a constant $c>1$. However, this is not always achievable when nodes have limited power. Suppose that each sender node has maximum power $P_{max}$. For concreteness, we assume that $c=2$. A link $i$ is called a \emph{weak link} if $P_{max} \le 2\beta N l_i^\alpha$.  Note that a link is weak because it is too long for its maximum power. In other terms, link $i$ is weak if $l_i \ge l_{max}/2^{1/\alpha}$, where  $l_{max}=(P_{max}/\beta N)^{1/\alpha}$ is the maximum length a link can have to be able to overcome the noise when using maximum power. Scheduling weak links may be considered as a separate problem. Let $\tau$-{\wscheduling} denote the problem of scheduling weak links with power scheme $P_{\tau}$. For a weak link $i$, let us call $e_i=c_i^{1/\alpha}l_i$ the \emph{effective length} of link $i$ and let $\Delta_e(S)=\max_{i,j\in S}{e_i/e_j}$. One approach to {\wscheduling} is to split the set of weak links into $O(\log{\Delta_e})$ effective-length classes $S'$ with $\Delta_e(S')\le 2$ and solve each of the classes using known algorithms for non-weak scheduling, thus achieving $O(\log{\Delta_e})$ approximation. Unfortunately, no better than $O(\log{\Delta_e})$ or $O(\log{n})$ approximation algorithm is known. The following theorem shows that constant factor approximation of {\wscheduling} is at least as hard as constant factor approximation of {\scheduling} with fixed uniform power scheme (denoted {\uscheduling}).
\begin{theorem}\label{T:weakhard}
There is a polynomial-time reduction from {\uscheduling}
to $\tau$-{\wscheduling} for any $\tau\in [0,1)$, transforming an arbitrary set $L$ of links to a set $W$ of weak links so that the scheduling number of $L$ with $P_0$ is within a constant factor of the scheduling
number of $W$ using $P_{\tau}$.
\end{theorem}

\begin{proof}
\newcommand{\pmax}{P_{max}}
\newcommand{\lmax}{l_{max}}
\newcommand{\lhat}{\hat{l}}

We consider links in the 2D plane -- similar arguments work for higher dimensions.

%The idea is to map the links in $L$ to weak links by appropriate scaling. 
The length $\lhat = \lmax/2^{1/\alpha}=(\pmax/(2\beta N))^{1/\alpha}$ is the
border between non-weak and weak links in $L$ (representing also the link
length for which $c_i=2$). Let $l_{min}$ be the length of the shortest link in $L$.
We want to map the links with length in range $[l_{min}, \lmax)$ to the range
$[\lhat, \lmax)$ in a manner that preserves the scheduling number, modulo constant factors.

Let $c:(0,\lmax)\rightarrow (\beta,\infty)$ be the
function that produces the coefficients $c_i$ in the
affectance definition with uniform power, defined by $c(x) =
\frac{1}{1 - \beta N x^\alpha/\pmax} 
= \frac{1}{1 - (x/\lmax)^\alpha}$;
i.e., for a link $i$ using uniform power, $c(l_i) = c_i$.
Let $e:(0,\lmax)\rightarrow (0,\infty)$ be the function
converting link lengths to effective lengths, defined by $e(x) =
c(x)^{1/\alpha} \cdot x$. Note that $e$ is monotone increasing and therefore invertible;
let $f(x) = e^{-1}(x)$ denote its inverse. Note that $f$ is monotone increasing and sublinear.

We now describe the instance $W$ created from $L$, using the
same parameters $\alpha, \beta, N$.
Let $X = 2^{1/\alpha} \lhat/l_{min}$.
For each $i\in L$, there is a link $i'\in W$ of length $l_{i'} = f(X l_v)$,
such that $s_{i'} = s_i \cdot X$ and $r_{i'} = s_{i'} + (l_{i'}, 0))$. 
This completes the construction.
The idea with the construction is that affectances involving $l_i$
and $l_{i'}$ should be essentially the same, measured at
the senders, and not differ too much at the receivers.

First, we verify that the construction forms only weak links.
%For comparison, let $\hat{i}$ denote a hypothetical link of length $\lhat$.
For each $i$, it holds that $X l_i \ge 2^{1/\alpha} \lhat = c(\lhat)^{1/\alpha} \lhat = e(\lhat)$,
so $l_{i'} = f(X l_i) \ge f(e(\lhat)) = \lhat$.
Also, $f(x) < \lmax$, for any $x \in \mathbb{R}$.
Hence, all link lengths are in $[\lhat,\lmax)$.

Next, we relate the impact of different
power schemes on the affectance of weak links. Let $U = P_0$ be the constant function representing uniform power.

\begin{claim}
Consider any $0 \le \tau < 1$.
Then, for any pair of \emph{weak} links $i, j$, $a_U(i,j) = \Theta(a_{P_{\tau}}(i,j)).$
\end{claim}

\begin{proof}
Let $t_j=(l_j/\lmax)^\alpha$. Since the links are weak, we have $t_j\in [1/2, 1)$ and $P_{\tau}(i) = \Theta(P_{\tau}(j))$. 
Recall, $P_{\tau}(j) = \pmax (l_j/\lmax)^{\tau\alpha} = t_j^{\tau} \pmax$ and $\pmax=\beta N \lmax^\alpha$.
Then, 
\[ c_j^U = c(l_j) = \frac{1}{1 - (l_j/\lmax)^\alpha} = \frac{1}{1 - t_j} 
  \quad \mbox{and} \quad
   c_j^{P_{\tau}} = \frac{1}{1 - \beta N l_j^\alpha/P_{\tau}(j)} 
     = \frac{1}{1 - t^{1-\tau}_j} \ .  \]
Noting that the function $h(x) = (1-x^{c})/(1-x)$, for $c \in (0,1)$,
satisfies 
$\lim_{x\rightarrow 1} h(x) = c$, and is bounded by a constant away from 0 elsewhere,
we get that $c_j^U = \theta(c_j^{P_{\tau}})$.
It follows that
 \[ a_{P_{\tau}}(i,j) = c_j^{P_{\tau}} \left(\frac{l_j}{d_{ji}}\right)^\alpha
    \cdot \left(\frac{P_{\tau}(i)}{P_{\tau}(j)}\right)^\alpha
 = \Theta(c_j^U) \left(\frac{l_j}{d_{ji}}\right)^\alpha
    \cdot \Theta(1)
 = \Theta(a_U(i,j)\ . \]
\end{proof}
Due to the last claim, for completing the proof, it is enough to show that the scheduling number of $L$ using $U$ is at most constant factor away from the scheduling number of $W$, again using $U$.
Let $S$ be a $U$-feasible subset of $L$ and $T'$ be a $U$-feasible subset of $W$. Also, let $S'=\{i'\in W : i \in S\}$ and $T=\{i\in L : i'\in T'\}$. We show that $S'$ and $T$ can be split into a constant number of $U$-feasible subsets w.r.t. $U$. 

Let us start with $S'$. First, observe that for any pair of links $i,j\in S$,
\begin{equation}
 a_U(i,j) = c_{j} (l_{j}/d_{ij})^\alpha
\label{eqn:avw}
\end{equation}
and
\begin{equation}
 a_U(i',j') = c_{j'} (l_{j'}/d_{i'j'})^\alpha
   = e(l_{j'})^\alpha/d_{i'j'}^\alpha 
   = c_j ( X l_{j}/d_{i'j'})^\alpha \ .
\label{eqn:avwprime}
\end{equation}
Since $S$ is $U$-feasible, we can split it into a constant number of subsets, each $3^\alpha\beta$-$U$-feasible, using Theorem~\ref{T:signalstrengthening}. Let $S_1$ be one of those subsets and $S'_1\subseteq S'$ be the corresponding subset in $W$. It suffices to show that $S'$ is $U$-feasible. Let $i,j\in S$. Then, $c_{j} (l_{j}/d_{ij})^\alpha=a_U(i,j) \le 1/(3^\alpha \beta)$, which implies that $d_{ij} \ge 3l_j$.
By the triangle inequality, $d(s_i,s_j)\ge d_{ij} - l_j \ge 2d_{ij}/3\ge 2l_j$. Then,  $d(s_{i'}, r_{i'}) \ge 2Xl_{j}\ge 2f(Xl_j)=2l_{j'}$, by $d(s_{i'},s_{j'})=Xd(s_i,s_j)$ and sublinearity of $f$. Using the triangle inequality, 
\begin{equation}\label{E:drelation}
d_{i'j'} \ge d(s_{i'}, r_{i'}) - l_{j'}\ge d(s_{i'}, r_{i'})/2=Xd(s_i,s_j)/2 \ge Xd_{ij}/3.
\end{equation}
Using (\ref{eqn:avwprime}) and (\ref{E:drelation}), we have:
\[
a_U(S', j')=\sum_{i'\in S'}{c_j\left(\frac{Xl_j}{d_{i',j'}}\right)^\alpha}\le 3^\alpha \sum_{i'\in S'}c_j\left(\frac{Xl_j}{Xd_{ij}}\right)^\alpha=3^\alpha a_U(S,j)\le 1/\beta,
\]
for any $j'\in S'$, which means that $S'$ is $U$-feasible. 

A symmetric argument applies for the sets $T$ and $T'$. This completes the proof.
\end{proof}

\section{A Note on Distributed Scheduling Algorithms} \label{sec:distributed}

The centralized algorithm for computing the schedules from Thm.~\ref{T:schapprox} boils down to finding a vertex coloring in a constant-simplicial graph $\cG_{\gamma}^{\delta}$. The latter is done by coloring the links greedily in decreasing order by length, i.e. a link gets the first color not yet used by its neighbors in $\cG_{\gamma}^\delta$ (see e.g.~\cite{yeborodin}). Here we sketch a way to compute the schedules distributively. The main idea is to split the computation into $\log{\Delta}$ stages, where in each stage only a single length class $L_t$ acts and the others are silent. The links use uniform power, proportional to the maximum link length in the class. The computation starts from the class $L_0$, containing the longest link. First, a subroutine for computing constant factor schedules for equilength sets (e.g.~\cite{BilelDerbel,Deltaplus1}) is run for $L_t$. As soon as the links in $L_t$ establish a coloring, they broadcast their colors with uniform power using a local broadcast algorithm (e.g.~\cite{HMLocalBroadcast}). This way, the links in $L_t$ notify their shorter neighbors in $\cG_{\gamma}^{\delta}(L)$ about their color. Then, length class $L_{t+1}$ proceeds with the coloring algorithm, and so on.

There are many details to be taken into account for implementing and evaluating the algorithm, which depend on exact assumptions and model characteristics, so the analysis below should be taken with a grain of salt. The algorithms from~\cite{BilelDerbel,Deltaplus1} for coloring equilength sets of links run in time $O(opt_t\log{n})$, where $opt_t$ denotes the optimum schedule length for $L_t$. It is known that $opt_t=\Theta(D(\cG_{\gamma}(L_t)))$ for a constant $\gamma$, where $D(G)$ denotes the maximum degree of graph $G$ (see~\cite{us:talg12}). Even though these algorithms compute a coloring from scratch, we believe it is possible to adapt them to take into account the set of colors used by earlier links, without degrading the runtime significantly.
The local broadcast subroutine takes $O(D(\cG_{\gamma}(L_t)) + \log^2{n})$ rounds (with collision detection \cite{HMLocalBroadcast}; 
$O(D(\cG_{\gamma}(L_t))\log n + \log^2{n})$ rounds without it \cite{Yu12,HMLocalBroadcast}), as the contention happens only between the links in $L_t$. Thus, realization of this idea would give a distributed computation of schedules in time $O(\log{n}\cdot\sum_{t}{opt_t}+\log^2{n}\log{\Delta})=O(opt\cdot\log{n}\log{\Delta} + \log^2{n}\log{\Delta})$, where $opt$ is the optimum schedule length.

\mypara{Acknowledgements} We thank Christian Konrad for discussions that led to the results in Sec.~\ref{sec:knownalgos}.

\clearpage

\bibliographystyle{abbrv}
\bibliography{Bibliography}

\end{document}